\documentclass[11pt,a4paper]{llncs}
\usepackage{fullpage}

\usepackage{geometry}
\usepackage{graphicx}
\usepackage{amssymb}
\usepackage{epstopdf}
\usepackage{authblk}
\usepackage{wrapfig}
\usepackage{subfig}
\usepackage{hyperref}
\usepackage{amsmath}
\usepackage{appendix}
\usepackage{enumerate}
\usepackage{delarray}
\usepackage{stmaryrd}
\usepackage{tikz}
\usetikzlibrary{fadings,decorations,decorations.pathreplacing}
\newcommand{\N}{\mathbb{N}}

\title{Avalanche Structure in the Kadanoff Sand Pile Model
\thanks{Partially supported by  IXXI (Complex System Institute, Lyon) and ANR project  Subtile. }}

\author{Kevin Perrot \and Eric R\'emila}
\institute{Universit\'e  de Lyon\\
  Laboratoire de l'Informatique du Parall\'elisme, \\
  (umr 5668 CNRS - ENS Lyon - Universit\'e  Lyon 1),\\
46 all\'ee d'Italie 69364 Lyon Cedex 7 - France,\\
\email{\{kevin.perrot,eric.remila\}@ens-lyon.fr }}

\begin{document}

\maketitle

\begin{abstract}
  Sand pile models are dynamical systems emphasizing the phenomenon of {\em Self Organized Criticality} (SOC). From $N$ stacked grains, iterating evolution rules leads to some critical configuration where a small disturbance has deep consequences on the system, involving numerous steps of grain fall. Physicists L. Kadanoff {\em et al} inspire KSPM, a model presenting a sharp SOC behavior, extending the well known {\em Sand Pile Model}. In KSPM with parameter $D$ we start from a pile of $N$ stacked grains and apply the rule: $D\!-\!1$ grains can fall from column $i$ onto the $D\!-\!1$ adjacent columns to the right if the difference of height between columns $i$ and $i\!+\!1$ is greater or equal to $D$. We propose an iterative study of KSPM evolution where one single grain addition is repeated on a heap of sand. The sequence of grain falls following a single grain addition is called an avalanche. From a certain column precisely studied for $D=3$, we provide a plain process describing avalanches.\\
  

\end{abstract}

\section{Introduction}

Sand pile models were introduced in \cite{bak88} as systems presenting a critical self-organized behavior, a property of dynamical systems having critical points as attractors. In the scope of sand piles, starting from an initial configuration of $N$ stacked grains the local evolution of particles is described by one or more iteration rules. Successive applications of such rules alter the configuration until it reaches an attractor, namely a stable state from which no rule can be applied. SOC property means those attractors are critical in the sense that a small perturbation | adding some more grains | involves an arbitrary deep reorganization of the system. Sand pile models were well studied in recent years (\cite{goles93},\cite{durandlose98},\cite{formenti07},\cite{phan08}).

\subsection{Kadanoff sand pile model}

  In \cite{kadanoff89}, Kadanoff proposed a generalization of classical models closer to physical behavior of sand piles in which more than one grain can fall from a column during one iteration. Informally, Kadanoff sand pile model with parameter $D$ and $N$ grains is a discrete dynamical system, which initial configuration is composed of $N$ stacks grains, moving in discrete space and time according to a transition rule : if the height difference between column $i$ and $i+1$ is greater or equal to $D$, then $D-1$ grains can fall from column $i$ to the $D-1$ adjacent columns on the right (see figure \ref{fig:rule}).

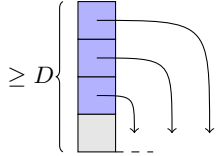
\begin{figure}[!h]
  \begin{center}
    \begin{tikzpicture}
  \foreach \y in {1,...,3}
    \filldraw[fill=blue!30] (0,.5*\y) rectangle ++ (.5,.5);
  \filldraw[fill=black!10] (0,0) rectangle ++ (.5,.5);
  \draw[dashed] (.5,0) -- ++ (.5,0);
  \draw[decorate, decoration=brace] (-.2,0) -- node [left] {$\geq D$} ++ (0,2);
  \draw[->] (.25,.75) .. controls (.75,.75) .. (.75,.25);
  \draw[->] (.25,1.25) .. controls (1.25,1.25) .. (1.25,.25);
  \draw[->] (.25,1.75) .. controls (1.75,1.75) .. (1.75,.25);
\end{tikzpicture}
  \end{center}
  \caption{KSPM($D$) transition rule.}
  \label{fig:rule}
\end{figure}

Sand pile models are specializations of {\em Chip Firing Games} (CFG). A CFG is played on a graph in which each vertex $v$ has a load $l(v)$ and a threshold $t(v)=deg^+(v)\footnote{$deg^+(v)$ denotes the out-degree of $v$.}$, and the transition rule is: if $l(v)\geq t(v)$ then $v$ gives one unit to each of its neighbors (we say $v$ is fired). As a consequence, we inherit all  properties of CFGs. 

Kadanoff sand pile is referred to a {\em linear chip firing game} in \cite{goles02}. The authors show that the set of reachable configurations endowed with the order induced by the successor relation has a lattice structure, in particular it has a unique {\em fixed point}. Since the model is non-deterministic, they also prove \emph{strong convergence} {\em i.e.} the number of iterations to reach the fixed point is the same whatever the evolution strategy is. The morphism from KSPM(3) to CFG is depicted on figure \ref{fig:lcfg}.

  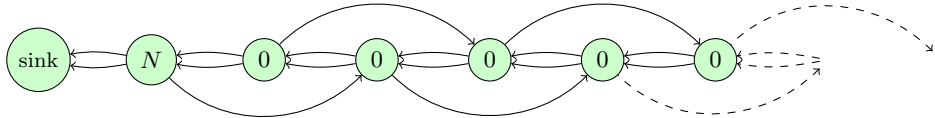
\begin{figure}[!h]
  \begin{center}
    \begin{tikzpicture}
  \node[circle, draw, fill=green!20] (n0) at (0,0) {$N$};
  \node[circle, draw, fill=green!20] (n-1) at (-1.5,0) {\scriptsize sink}
    edge [<-,out=10,in=170] (n0)
    edge [<-,out=-10,in=-170] (n0);
  \node[circle, draw, fill=green!20] (n1) at (1.5,0) {0}
    edge [->,out=170,in=10] (n0)
    edge [->,out=-170,in=-10] (n0);
  \node[circle, draw, fill=green!20] (n2) at (3,0) {0}
    edge [->,out=170,in=10] (n1)
    edge [->,out=-170,in=-10] (n1)
    edge [<-,out=-135,in=-45] (n0);
  \node[circle, draw, fill=green!20] (n3) at (4.5,0) {0}
    edge [->,out=170,in=10] (n2)
    edge [->,out=-170,in=-10] (n2)
    edge [<-,out=135,in=45] (n1);
  \node[circle, draw, fill=green!20] (n4) at (6,0) {0}
    edge [->,out=170,in=10] (n3)
    edge [->,out=-170,in=-10] (n3)
    edge [<-,out=-135,in=-45] (n2);
  \node[circle, draw, fill=green!20] (n5) at (7.5,0) {0}
    edge [->,out=170,in=10] (n4)
    edge [->,out=-170,in=-10] (n4)
    edge [<-,out=135,in=45] (n3);
  \node (n6) at (9,0) {}
    edge [->,dashed,out=170,in=10] (n5)
    edge [->,dashed,out=-170,in=-10] (n5)
    edge [<-,dashed,out=-135,in=-45] (n4);
  \node (n7) at (10.5,0) {}
    edge [<-,dashed,out=135,in=45] (n5);
\end{tikzpicture}
  \end{center}
  \vspace{-.5cm}
  \caption{The initial configuration $\sigma$ of KSPM(3) is presented as a CFG where each vertex corresponds to a column (except the sink) seen as a difference of height.}
  \label{fig:lcfg}
\end{figure}

  More formally, sand pile models we consider are defined on the space of ultimately null decreasing integer sequences. Each integer represents a column of stacked sand grains and transition rules describe how grains can move from columns. Let $h=(h_0,h_1,h_2,\dots)$ denote a {\em configuration} of the model, where each integer $h_i$ is the number of grains on column $i$. Configurations can also be given as height difference $\sigma=(\sigma_0,\sigma_1,\sigma_2,\dots)$, where for all $i \geq 0,~ \sigma_i=h_i-h_{i+1}$. We will use this latter representation throughout the paper, within the space of ultimately null non-negative integer sequences.

\begin{definition}
The   Kadanoff sand pile model with parameter $D$, KSPM($D$), is defined by:
  \begin{itemize}
    \item A set of \emph{configurations}, consisting in ultimately null non-negative integer sequences.
    \item A set of \emph{transition rules} : we have a transition from a configuration $\sigma$ to a configuration $\sigma '$ on column $i$, and we note   $\sigma \overset{i}{\rightarrow} \sigma'$ if  
    \begin{itemize}
\item $\sigma'_{i-1}=\sigma_{i-1} + D-1$ (for $i \neq 0$)
\item $\sigma'_i = \sigma_i - D$, 
\item $\sigma'_{i+D-1} = \sigma_{i+D-1} + 1$
\item $\sigma'_j = \sigma_j$ for $j \not\in  \{i-1, i, i+D-1 \}$. 
\end{itemize}

  \end{itemize}
\end{definition}

Remark that according to the definition of the transition rules, a condition for $\sigma'$ to be a configuration is that $\sigma_i \geq D$.

\subsection{Strategies and avalanches}

A basic property of the KSPM model is the \emph{diamond property}. If there exists two distinct integers $i$ and $j$ such that 
$\sigma \overset{i}{\rightarrow} \sigma'$ and $\sigma \overset{j}{\rightarrow} \sigma''$, then there exists a configuration $\sigma'''$ such that $\sigma' \overset{j}{\rightarrow} \sigma'''$  and $\sigma'' \overset{i}{\rightarrow} \sigma'''$. 
We  note $\sigma \rightarrow \sigma'$ when there exists an integer $i$ such that $\sigma \overset{i}{\rightarrow} \sigma'$. 
We define the transitive closure $\overset{*}{\rightarrow}$ of $\rightarrow$, and say that $\sigma'$ is {\em reachable} from $\sigma$ when  $\sigma \overset{*}{\rightarrow} \sigma'$.

A {\em strategy} is a sequence $s=(s_1,\dots,s_T)$. We say that $\sigma'$ is {\em reached} from $\sigma$ via $s$ when $\sigma \overset{s_1}{\rightarrow} \sigma'' \overset{s_2}{\rightarrow} \dots \overset{s_T}{\rightarrow} \sigma'$ and we note $\sigma \overset{s}{\rightarrow} \sigma'$. 
We also say,  for each integer $t$ such that $0 < t \leq T$, that the column $s_t$ \emph{is fired} at \emph{time} $t$ in $s$. (informally,  the index of the sequence is interpreted as time). 

 For any strategy $s$ and any nonnegative integer $i$, we state $|s|_i=\#\{ t | s_t = i \}$.  Let   $s^0$, $s^1$ be  two strategies such that $\sigma \overset{s^0}{\rightarrow} \sigma^0$ and $\sigma \overset{s^1}{\rightarrow} \sigma^1$.  We have the equivalence:  $[\forall~ i, |s^0|_i = |s^1|_i] \Leftrightarrow \sigma^0 = \sigma^1$.
 A strategy $s$ such that $\sigma \overset{s}{\rightarrow} \sigma'$ is called {\em leftmost} if it is the minimal strategy from $\sigma$ to $\sigma'$ according to lexicographic order. A leftmost strategy is such that at each iteration, the leftmost possible transition is performed. 

We say that a configuration $\sigma$ is \emph{stable}, or a \emph{fixed point} if no transition is possible from $\sigma$. 
As a  consequence of the diamond property, one can easily check that, for each configuration $\sigma$, there exists a unique stable configuration, denoted by $\pi(\sigma)$, such that  $\sigma \overset{*}{\rightarrow} \pi(\sigma)$. Moreover,  for any  
configuration $\sigma '$ such that $\sigma \overset{*}{\rightarrow} \sigma '$, we have $\pi(\sigma') = \pi(\sigma)$ (see \cite{goles02} for details). 

In this paper, we are interested in the iterative process defined below. Starting with no grain, we successively add a single  grain on  column 0, and make all the possible firings until  a fixed point is reached. We denote by $\pi(k)$ the configuration obtained with this process using $k$ grains (from the structure of KSPM described above, one easily checks that $\pi(k) = \pi((k,0^\omega)$).

Let $\sigma$ be a configuration, $\sigma^{\downarrow 0}$ is the configuration obtained by adding one grain on column $0$. In other words, if $\sigma=(\sigma_0
,\sigma_1,\dots)$, then $\sigma^{\downarrow 0}=(\sigma_0 +1 ,\sigma_1,\dots)$. 

Formally the process is defined by $\pi (0) = 0^\omega$ and the recurrence formula: 
$$\pi(\pi(k-1)^{\downarrow 0})  =   \pi (k). $$

The {\em $k^{th}$ avalanche} $s^k$ is  the leftmost strategy from $\pi(k-1)^{\downarrow 0}$ to $\pi(k)$. 
The goal of the present paper is the description of avalanches. Informally, we want to describe what happens when a new grain is added in a previously stabilized sand pile.

For $D = 2$, i.e.  the classical SPM, this description is easy: the added grain moves rightwards until it arrives in a plateau. But, for $D > 2$,  the situation is not so simple. We now state our results. 

 \begin{itemize}
\item In the general case, we prove (Section \ref{section:avalanche}) the following properties:  
 \begin{itemize}
\item Each column is fired at most once, 
\item For any avalanche, as soon as an interval  $\{ L, L+1, ...., L+D-1 \} $ of successive fired columns exists, the execution of the avalanche on the right part of this interval can be turned into a pseudo local  and elementary process. 

\end{itemize}
Informally, that means that the knowledge of such an interval guarantees a regular behavior  of the avalanche on its right part. 

\item  In the case when $D = 3$,  we prove (Section \ref{section:bounding}) the property below: 
\begin{itemize}
\item For each avalanche $s^k$, there  exists an integer  $L(k)$ in $O(\log{k})$ such that either no column is fired on the right of $L(k)$,  or  columns $L(k)$ and $L(k) +1$ both are fired (and therefore, the property of the second item above applies). 
\end{itemize}
Informally, that means that we have the emergence of  a regular behavior, after a short transitional and complex phase. 
\end{itemize}

These results give a better understanding of avalanches for sufficiently large columns. We hope that in future work, they will help us in the  approach of the structures of fixed points $\pi(k)$. 

\subsection{The context}

 The problem of describing and proving regularity properties, experimentally checked,  for   models issued from  basic dynamics is really a present challenge for physicists, mathematicians, and computer scientists. There exists
a lot of conjectures, issued from simulations, on  discrete dynamical systems with simple local rules (sandpile model \cite{dartois} or chip firing games, but also  rotor router  \cite{levine},  the famous Langton ant\cite{gajardo}\cite{propp}...)  but very few results have actually been proved. As regards KSPM($D$), the {\em prediction problem} (namely, the problem of computing the fixed point $\pi(k)$ knowing $\pi(k-1)$) has been proven in \cite{moore99} to be in \textbf{NC}$^3$ for the one dimensional case (the model of our purpose), which means that the time needed to compute an avalanche is in $O(\log^3 N)$ where $N$ is the number of grains, and \textbf{P}-complete when the dimension is $\geq 3$.  A recent study (\cite{goles10}) showed that in the two dimensional case the avalanche problem (given a configuration $\sigma$ and a column $i$ on which we add one grain, does it have an influence an index $j$?) is \textbf{P}-complete, which points out a inherently sequential behavior.

This study will provide tools to understand sand pile evolution. We hope that those tools form a basis to obtain  some good descriptions of fixed points $\pi(k)$, but are also deeply related with other subjects around sand piles such as unit elements of abelian group structures presented 
in \cite{creutz96} and \cite{dhar90}.

\section{Avalanche process in the general case}\label{section:avalanche}

This section begins with a first glance at avalanches, allowing notation simplifications. Then avalanches are studied in details, leading to a simplified description of its behavior.

\begin{proposition}\label{lemma:01}
  For each strategy $s$  such that  $\pi(N)^{\downarrow 0}  \overset{s} {\rightarrow}\pi(N+1)$ and  each $i \in \N$ , we have  $|s|_i \in \{0,1\}$. 
\end{proposition}

\begin{proof}
  Let $s=(s_1,\dots,s_T)$ be a strategy such that  $\pi(N)^{\downarrow 0}  \overset{s} {\rightarrow}\pi(N+1)$. We have to  prove  that, for  
  $1 \leq l < m \leq T$, we have $s_l \not = s_m$ (obviously, $|s|_i \geq 0$ for all $i$). To do it,  we prove by induction $t \leq T$ that for
$1 \leq l < m \leq t$, we have $s_l \not = s_m$.\\ 
  For initialization this is obviously true for $t = 1$. 
   Now assume that the condition is satisfied for an integer $t$ such $t < T$, and let $i$ be a column such that there exists an integer $l \leq t$ such that $i = s_l$. Let $\sigma$  be the  configuration such that $\pi(N)^{\downarrow 0} \overset{s_1}{\rightarrow} \dots \overset{s_t}{\rightarrow} \sigma$.\\
  Notice that the transitions which can possibly change the value of the current configuration at $i$ could be: $i$ (which decreases the value by $D$ units), $i+1$ (which increases the value by $D -1$ units) or  $i - D+1$ (which increases the value by $1$ unit).\\
  Thus we have $\sigma_i \leq  \pi(N)^{\downarrow 0}_i - D + D-1 +1$ since by definition,  between   $\pi(N)$ and $\sigma$, exactly one transition has occurred in   $i$, at most one transition has occurred in $i+1$, and at most one transition has occurred in $i - D+1$. For $i \geq 1$, we get $\sigma_i \leq  \pi(N)_i $. On the other hand, since $\pi(N)$ is a fixed point, we have:   $\pi(N)_i < D $,  which guarantees that $s_{t+1} \neq i$. For $i = 0 $, there is no possible transition in $i - D+1$, thus we get $\sigma_0 \leq  \pi(N)^{\downarrow 0}_0 - D + D-1 $, which is  $\sigma_0 \leq  \pi(N)_0 +1 - D + D-1$. Thus    $\sigma_0 \leq  \pi(N)_0 < D $ which also gives: $s_{t+1}  \neq 0$.\\
  This ensures that the result is true for $t+1$,  and,  by induction,  for $T$. 
\end{proof}

When talking about an avalanche $s$, lemma \ref{lemma:01} allows us to write $i \in s$ instead of $|s|_i=1$ without lose of information. We denote by  $s_{[u, v]}$  the  subsequence of $s$ from  $u$ to $v$ included.

We will now study avalanches in details. For $D = 2$, i.e.  the classical SPM, avalanches are quite simple, the added grain moves rightwards until it finds a stable position. For $D > 2$, the situation is more complex, and needs a precise study, given by the following lemma.

\begin{lemma}\label{lemma:localdensity}
Let $s=(s_1,\dots,s_{t_k})$ be the $k^{th}$ avalanche. Let $r_t = max \{s_{t'}, t' \leq t \}$.
\begin{itemize}
\item Assume that $s_{t+1} < r_t$.  Then $s_{t+1}$  is the largest column number $i$ such that $i < r_t$ and $i \notin s_{[1,t]}$. Moreover,  we have: $r_t - s_{t+1} < D-1$.
\item Assume that $s_{t+1} > r_t$. We have:  $s_{t+1} -r_t  \leq  D-1$.
\end{itemize}

\end{lemma}

\begin{proof}
We first order fired columns by causality. Precisely, a column $i$ has  two potential predecessors, which are
$i +1$ and   $ i - D+1$. State $i = s_u$. These columns are really predecessors of $i$ if they are elements of $s_{[1, u ]}$, i.e if they are fired before $i$. By this way, using the transitive closure,  we define a partial order relation (denoted by $<_{caus.}$),  on fired columns for $s$.

Now,  consider the set $A_{t+1 }$ of ancestors of $s_{t+1 }$ (i.e.  the set of columns $i$ such that $i <_{caus.} s_{t+1}$)  and,    the set $S_t$ of which have  $s_t$ as common ancestor (i.e. columns $i$ such that $s_t  <_{caus.} i$).
We necessarily have  $r_t \in A_{t+1 }$.   Otherwise, we have  $A_{t+1 } \cap S_t =  \emptyset$, and this  allows  another strategy $s'$, constructed from $s$ postponing the transitions at  $r_t$ and elements of $S_t$ after the transition on  $s_{t+1}$. This contradicts the fact that $s$ is leftmost.

Let  $(i_0, i_1, ...i_p)$ be a finite sequence such that $i_0 = r_t$, $i_p = s_{t+1}$ and,  for each $j$ with $0 \leq j < p$, $i_j$ is a predecessor of $i_{j+1}$. Such a sequence exists since $r_t \in A_{t+1 }$. One easily proves by induction that $i_j = r_t -j$: 
this is true for $i = 0$. Assume it is true until the integer $j < p$. We have either $i_{j+1} = i_j -1$ or  $i_{j+1} = i_j +D -1$. But from the induction hypothesis, $i_j +D -1$ is an ancestor of $i_{j+1}$ or has not yet been fired, thus $i_{j+1} = i_j -1$.  This gives that   $s_{t+1}$  is the largest column number $i$ such that $i < r_t$ and $i \notin s_{[1,t]}$. 

Now if we assume, by contradiction, that  $p \geq  D-1$, then  $r_t -D+1$ is not a predecessor of $r_t$, which yields that $r_t$ has no predecessor, which is a contradiction. This gives the  inequality   of the first item. 
The second item  is obvious, since $s_{t+1}$ ha a unique predecessor which is $s_{t+1}-D+1$. 
\end{proof}

Lemma \ref{lemma:localdensity} induces a partition of fired columns between those which make a progress ({\em i.e.} increases the greatest fired column) and those which do not. This distinction is important in further development, so let us give progress firings a name. Let $s=(s_1,\dots,s_T)$ be an avalanche, a column $s_{t}$ is called a {\em peak} if and only if $s_{t} > \max s_{[1,t-1]}$.

\begin{remark}\label{remark:order}
  Two peaks $p\not =q$ can be compared using chronological ($<_T$) or spatial ($<_S$) orders. Nevertheless, by definition of peaks we obviously have $p <_T q \iff p <_S q$.
\end{remark}

\begin{lemma}\label{lemma:D-1}
  Let $s$ be the $k^{th}$ avalanche. Assume that  there exists a column  $l$,  such that for each column  $i$ with $l \leq i < l+D-1$, $i \in s$, and a column $i'$ such that $i'\geq l+D-1$ and $i' \in s$. Let  $l'$ be  the lowest peak such that  $l' \geq l+D-1$.
  
There exists a time $t$  such that:

 \begin{tabular}[t]\{{l}. 
 for all $i$ with $l' -D+1 < i \leq  l'$, $i \in s_{[1,t]}$\\
  for all $i$ with $l'< i$, $i \notin s_{[1,t]}$
   \end{tabular}
   
  Moreover $l'$ is the lowest integer such that  $l' \geq l + D-1$ and $\sigma^t_{l'} = D-1$. 
\end{lemma}

\begin{proof}

 Let $t_0$ be the time when $s_{t_0} = l'$, i.e. the   first  time such that $s_{t_0} \geq l+D-1$,  and let $j$ be the largest integer such that,  for $0 \leq j' \leq j$, we have $s_{t_0 +j'} = s_{t_0} -j'$.  Let us state $t = t_0 +j$. We have $j < D-1$. \\
Let $\sigma^t$ denote the configuration obtained from $\pi(k-1)$ via $s_{[1,t]}$. Let   $i$, with $ i < l'$,  such that $i \notin s_{[1,t]}$. We claim that we have :  $i \notin s$. To prove it, we prove by induction that for any $t' \geq t$,  $i \notin s_{[1,t']}$. Assume that this is satisfied for a fixed $t'$. This means that all the transitions of $s_{[t+1,t']} $  are done on columns larger than $l'$.  Thus, $ \sigma^{t'}_i =   \sigma^{t}_i$ and no transition is possible  on $i$ for $\sigma^{t}$ since  $s$ is leftmost (the only potential column to be fired is $ s_{t_0} -j -1$, but by assumption, either this column has been previously fired,  or it cannot be fired by definition of $j$, according to lemma \ref{lemma:localdensity}) . \\
By contraposition,  it follows that for each column  $i$ with $l \leq i < l+D-1$,  we have $i \in s_{[1,t]}$.  A   simple (reverse sense) induction shows that,   for $ l+D-1 \leq i \leq l'$   we have $i \in s$, since  by hypothesis $i+1$, and $i+1 -D$ both are in $s$.  Thus, by contraposition of the claim above,  for $ l+D-1 \leq i \leq l'$,   we have $i \in s_{[1,t]}$. 
This gives  the the fact that for all $i$ with $l' -D+1 < i \leq  l'$, $i \in s_{[1,t]}$\\
The fact that for all $i$ with $l'< i$, $i \notin s_{[1,t]}$ is trivial,  by definition of $t_0$ and $t$.\\
 We have $l' > l + D-2$ and $\sigma^t_{l'} = D-1$. assume that there exists $l''< l'$ satisfying the same properties. Notice that 
 for $t_0 \leq  t' \leq t$   we have $s_t' > l' - D-1$. Thus the time $t_1$ such that $s_{t_1} = l'' -D+1$  is such that $t_1 < t_0$. 
  That means that $ l'' $  should have been fired before $t_0$,  a contradiction.
\end{proof}

Lemma \ref{lemma:D-1} describes in a very simple way the behavior of avalanches. Thank to it, the study of an avalanche can be turned into a pseudo linear execution, in which transitions are organized in a clear fashion:

\begin{theorem}\label{corollary:peak}
  Let $s=(s_1,\dots,s_T)$ be the $k^{th}$ avalanche and $(p_1\dots,p_q)$ be its sequence of peaks. Assume that  there exists a column $l$,  such that for each column  $i$ with $l \leq i < l+D-1$, $i \in s$. 
  Then for any column $p$ such that $p \geq l+D-1$, 
  $$p\text{ is a peak of }  s  \iff  p \leq p_q +D-1\text{ and }\pi(k-1)_p = D-1$$
  
  Furthermore, Let $p_i=s_t$, with $p_i \geq l+D-1$, be a peak. Then
  $$T \geq t+p_i-p_{i-1}-1 \text{ and for all } t' \text{ s.t. } t < t' \leq t+p_i-p_{i-1}-1,~ s_{t'}=s_{t'-1}-1$$

 \end{theorem}

 A graphical representation of this statement is given on figure \ref{fig:peak}.

\begin{proof}
  The first part  is a straight induction on lemma \ref{lemma:D-1}.\\
  The second part follows an induction summed up in the following fact: any column $i$ such that $\pi(k-1)_i < D-1$ must wait for its right neighbor $i+1$ to be fired, and it should be fired when both $i+1$ and $i-D+1$ has been fired (besides, $i-D+1$ has already been fired). Since any of such $i$ is fired to reach a fixed point, $T \geq t+p_i - p_{i-1} -1$.
\end{proof}

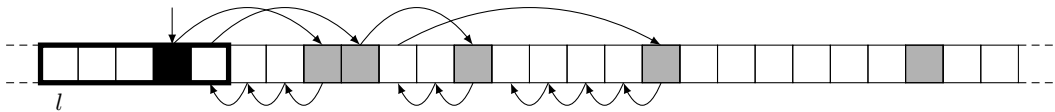
\begin{figure}
\begin{tikzpicture}
  \filldraw[fill=black] (3*.5,0) rectangle ++ (.5,.5);
  \foreach \x in {7,8,11,16,23}
    \filldraw[fill=black!30] (\x*.5,0) rectangle ++ (.5,.5);
  \foreach \x in {0,...,25}
    \draw (\x*.5,0) rectangle ++ (.5,.5);
  \draw[densely dashed] (13,0) -- ++ (.5,0);
  \draw[densely dashed] (13,.5) -- ++ (.5,0);
  \draw[densely dashed] (0,0) -- ++ (-.5,0);
  \draw[densely dashed] (0,.5) -- ++ (-.5,0);
  \draw[line width=2pt] (0,0) -- ++ (2.5,0) -- ++ (0,.5) -- ++ (-2.5,0) -- cycle;
  \node at (.25,-.25) {$l$};
  \draw[-latex] (3*.5+.25,1) -- ++ (0,-.5);
  \draw[-latex] (3*.5+.25,.5) parabola[bend pos=0.5] bend +(0,.5) +(4*.5,0);
  \draw[-latex] (4*.5+.25,.5) parabola[bend pos=0.5] bend +(0,.5) +(4*.5,0);
  \draw[-latex] (8*.5+.25,.5) parabola[bend pos=0.5] bend +(0,.5) +(3*.5,0);
  \draw[-latex] (9*.5+.25,.5) parabola[bend pos=0.5] bend +(0,.5) +(7*.5,0);
  \foreach \x in {5,6,7,10,11,13,14,15,16}
    \draw[-latex] (\x*.5+.25,0) parabola[bend pos=0.5] bend +(0,-.3) +(-.5,0);
\end{tikzpicture}
\caption{Illustration of Theorem \ref{corollary:peak} with $D=6$. Surrounded columns $l$ to $l+D-2$ are supposed to be fired. Black column is the greatest peak strictly lower than $L+D-1$. A column is grey if and only if its value is $D-1$. Following arrows depicts the avalanche.}
\label{fig:peak}
\end{figure}

  In the light of remark \ref{remark:order}, theorem \ref{corollary:peak} easily allows us to compute the right part of $k^{th}$ avalanche (from column $l + D-1$), only knowing $\pi(k-1)$. The sequence of peaks is computed as follows. The first one is the lowest column $i$ greater or equal to $l+D-1$ such that $\pi(k-1)_i=D-1$. Then, given a peak $i$, the next one is the lowest  $j$ such that $\pi(k-1)_j = D-1$  and $j- i \leq D-1$. If such a $j$ does not exist, then there is no more peak and $i$ is the largest fired column.

  We can distinguish two movements within an avalanche: before a certain column it has an unknown behavior, and from that column to the end the behavior is pseudo local, in the sense that when an index is fired ahead (on the right) then any `hole' is filled before the progress can continue.

An important direct implication of theorem \ref{corollary:peak} is that if there exists a column $l$ such that for the $k^{th}$ avalanche $s^k$, we have for all $l \leq i < l+D-1$, $i \in s$, then for all $j$ such that $l+D-1 \leq j < \max s^k$, we have $j-(D-1), j, j+1 \in s^k$ and therefore $\pi(k)_j=\pi(k-1)_j$. Intuitively, this equality hints some similarity between successive avalanches.

Note that previous results also apply for a grain addition on column 0 of any fixed point configuration of KSPM($D$).

This study constitutes a simplified understanding of the behavior of avalanches, which we hope will be helpful toward the description of fixed points. As motivated above, next subsection studies, for KSPM(3),  the previous result hypothesis  that for an avalanche $s$ there exists a column $l$ such that for all $l \leq i < l+D-1$, $i$ is element of $s$.\\
    
\section{Short transitional phase when $D=3$}\label{section:bounding}

In this section we prove that in KSPM(3), there exists a column $l(N)$ in $O(\log N)$ such that lemmas hypothesis is verified for any avalanche $s^k$, with $k \leq N$, such that $\max s^k > l(k)$. In other words, considering the $N$ first avalanches, from a logarithmic column,  we can apply theorem \ref{corollary:peak} and consider avalanches pseudo locally, as described on figure \ref{fig:peak}. Here is the statement:

\begin{proposition}\label{lemma:meta2}
Let $s$ be the $k^{th}$ avalanche of KSPM(3). There exists a column $l(k)$ in $O(log~ k)$ such that for any $k$, when  $\max\{j \vert j \in s \} > l(k)$, $l(k)$ and $l(k)+1$ both are elements of $s$.
\end{proposition}

\begin{proof}

Let $i$ be a fixed column. If $i,i+1 \in s^k$, then, theorem \ref{corollary:peak} states that for all $i'$ such that $i \leq i' < \max\{j \vert j \in s \}$, we have $i',i'+1 \in s$. If $j,j+1 \notin s$, then, from Proposition \ref{lemma:localdensity},  we have $\max\{j \vert j \in s^k \} < i$.

 Let $j$ be a fixed positive integer. Assume  that  the avalanche $s$ fires $2j $ but not $2j-1$. From remarks above, columns $0, 2, 4,..., 2j$ are fired in $s$ while  columns $1, 3, 5,..., 2j-1$ are not fired.

If a column $i+1$ is fired while $i$ is not, then we necessarily have $\pi(k-1)_i = 0$,  since the firing in $i+1$ increases the value in column $i$ from 2 units. 
Moreover, if the  column $i+1$ is fired while $i+2 $ is not, then we necessarily have $\pi(k-1)_{i +1}= 2$, since the $i+1$  receives at most one grain, by preceding firings. 

On the other hand,  obviously, the assumption on $j$ enforces that  $\pi(k-1)_0 = 2$ 
This yields that $(2,0)^{j-2}$ is a prefix  of $\pi(k-1)$. 

We have the following fact : \\ 

\textbf{Fact:} 
There exists constant numbers $A$ and $B$, with $A >0$ such that if a  configuration  $ \pi(N)$ has a prefix of the form $(2,0)^{j}$ then $N > A 4^j + B$\\
 
 This is obtained by  the  linear algebra analysis below. This gives the result of the proposition.

 Let $ \pi(N)  = (\sigma_0,\sigma_1,\dots)$ be the  configuration and $a=(a_0,a_1,\dots)$ be its 
 {\em shot vector} i.e.   the sequence  $a=(a_0,a_1,a_2,\dots)$ where $a_i$ is the number of times the column  $i$ has be fired in the $N$ first avalanches. According to the iteration rule we have the relation:
 
 $$\sigma_i=a_{i-2}-3a_i+2a_{i+1}$$
 i.e.   $$a_{i+1}=\frac{1}{2}(\sigma_i-a_{i-2}+3a_i)$$

We state  $A=\begin{pmatrix} 0 & 1 & 0\\ 0 & 0 & 1\\ -1/2 & 0 & 3/2 \end{pmatrix}$. We denote  by $v_i$ the column vector such that  and $v_i^T = ( 0,   0,  \sigma_i/2 ),$
and    $u_i$ the column vector such that $u_i^T =(a_{i-2},a_{i-1},a_i)$ (with the convention that  $u^T$ is  the row vector   obtained by transposition of the column vector $u$) . The equality above can be algebraically written in

 $$u_{i+1}=A u_i + v_i$$

By
iteration we get :

$$u_{i+2}=A^2u_i+A v_i + v_{i+1}$$

What we want is $\sigma_{2i}=2$ and $\sigma_{2i+1}=0$ so $v_{2i}^T=(0,0,1)$ and $v_{2i+1}^T=(0,0,0)$. With this specification, we get : 
$$ u_{2(i+1)}=A^2u_{2i}  + b $$

 with $b^T =  (0,1,3/2)$.
From this last relation we will deduce a condition on $N$ to get the sequence $(2,0)^j$.\\

Let us first find a new basis to get the matrix $A$ on Jordan canonical form. The characteristic polynomial of $A$ is $\frac{1}{2}(2x+1)(x-1)^2$ and its eigenvalues are $-\frac{1}{2}$ of algebraic multiplicity 1 and 1 of algebraic multiplicity 2. Since $\dim ( \ker (A-Id))=1$ the Jordan canonical form of A is
\[
 A_{Jordan} =\begin{pmatrix} 1 & 1 & 0\\ 0 & 1 & 0\\ 0 & 0 & -1/2 \end{pmatrix}
\]
And the new basis $E'=(e'_1,e'_2,e'_3)$ according to the canonical one $E=(e_1^T=(1,0,0), e_2^T=(0,1,0), e_3^T=(0,0,1))$ is given by the linear relations:
\begin{center}
 \begin{tabular}{l}
   $e'_1=e_1+e_2+e_3$\\
   $e'_2=e_2+2e_3$\\
   $e'_3=4e_1-2e_2+e_3$
 \end{tabular}
\end{center}

Let $p$ be  the
 linear mapping consisting in the projection on the line $D_3$ generated  by  $e'_3$ according to the direction the plane $P_{1, 2}$¬¨‚Ä† generated by $e'_1$ and $e'_2$. For any vector $u$, we have:  $p(A^2u) = p(A^2(p(u) + u -p(u))  = p(A^2(p(u)) +  p(A^2 (u -p(u))$. Notice that, by definition of $p$, $u -p(u) $ is element of $P_{1, 2}$, and,  therefore $A(u -p(u))$ and $A^2(u -p(u))$ also are elements of $P_{1, 2}$. This yields that $p(A^2 (u -p(u)))$ is null.
 On the other hand, since  $p(u) $ is element of $D_{3 }$, $A(p(u))  = \frac{-1}{2 }p(u)$ and $A^2(p(u))  = \frac{1}{4 }p(u)$;  thus  $p(A^2(p(u)) =  \frac{1}{4}p(u)$. As a conclusion, we get $p(A^2u) = \frac{1}{4}p(u)$, which, in particular, allows the following equalities:

\begin{center}
 \begin{tabular}{rl}
   $p(u_{2(i+1)})$ & $= p(A^2u_{2i}+b)$\\
    & $= p(A^2u_{2i}) + p(b)$\\
    & $= \frac{1}{4}p(u_{2i}) - \frac{1}{18}e'_3$
 \end{tabular}
\end{center}
Let $v$ be the unique vector collinear with $e'_3$ satisfying the equation
\[
 v=\frac{1}{4}v - \frac{1}{18}e'_3
\]
i. e.  $v=  \frac{-2}{27}e'_3$.  Remember that $p(v) = v$. We have
\begin{center}
 \begin{tabular}{rl}
   $p(u_{2(i+1)}-v)$ & $= p(u_{2(i+1)}) - v$\\
    & $= \frac{1}{4}p(u_{2i}) - \frac{1}{18}e'_3 -(\frac{1}{4}v - \frac{1}{18}e'_3) )$\\
    & $= \frac{1}{4}p(u_{2i}-v)$

 \end{tabular}
\end{center}
This gives by induction: 

$$ p(u_0 - v)  = 4^j p(u_{2j}-v)$$

Now we specify the sequence of vectors $u_i$,  assuming that values $a_i$ are the shot vectors of a configuration $\sigma $ beginning by $(2,0)^j$.
 (For convention we also  state $a_{-2}=N$ and $a_{-1}=0$, thus we have $u_0^T = (N, 0, a_0)$, $u_1^T = (0, a_0, a_1)$ and $u_i^T =(a_{i-2},a_{i-1},a_i)$ for $i \geq 2$).

 An easy computation gives that:  
 $p(u_0 - v) = \frac{N+ a_0 + \frac{2}{27}}{9}e'_3 $

 Let $x_j$ be defined by $p(u_{2j}-v) = x_j e'_3$. We obtain   the equality :

$$ \frac{N+ a_0 + \frac{2}{27}}{9} =  4^j x_j$$

Obviously $a_0 \leq \frac{N}{D} = \frac{N}{3}$, which ensures that  $N+ a_0 + \frac{2}{27}\leq \frac{4N}{3}  +1$.

Furthermore,  we necessarily have $x_j >0$, and, from lemma \ref{lemma:alpha} proved on the bounce, each element of $p(\mathbb Z^3)$ is a multiple of $c e'_3$, where $c$ is a positive constant.
If $v$  is element of  $p(\mathbb Z^3)$ we can conclude that $x_j \geq c.$
If $v$  is not element of  $p(\mathbb Z^3)$,  we can conclude that $x_j \geq   min\{ \vert ck + \frac{2}{27}\vert , k \in \mathbb{Z}  \}$. In any case,  there exists a positive real $d$, not depending on $j$,   such that  $x_j \geq d.$

We conclude that $\frac{4N}{3}  +1 \geq  4^j d $, which gives $N>  \frac{3d }{4} \,4^j -\frac{3 }{4} $ to get a sand pile of the form $\sigma=(2,0)^j\sigma'$.
\end{proof}

We now give the lemma used in the previous proof. 

\begin{lemma}(constant steps)\label{lemma:alpha}
  There exists a  positive real $c$ such that $p_{e'_3}(\mathbb Z^3) = \{ i \, c  e'_3, i \in \mathbb Z \}$.
\end{lemma}

\begin{proof}
  The set  of reals $r$ such that there exists an element $x$ in $\mathbb Z^3$ such that  $p_{e'_3}(x)  = r e'_3$ is obviously a group. So we only have to prove that this group is discrete, {\em i.e.} that there is no sequence $(r_n)_{n \in  \mathbb Z}$ of positive reals such that $\lim \limits_{n \rightarrow \infty} (r_n)= 0$.\\
 Assume, by contradiction, the existence of such a sequence, and let  $(x_n)_{n \in  \mathbb Z}$ be a sequence of vectors such that, for each integer $n$, $p_{e'_3}(x_n) = r_n$. \\
 A key-point is that vectors $e'_1$, $e'_2$ and $e'_3$ have  integer  components, so we can state $x_n = a_n e'_1 + b_n e'_2 + c_n e'_3$. 
 The sequence  $(x'_n)_{n \in  \mathbb Z}$ defined by  $x'_n = (a_n -  \lfloor a_n \rfloor) e'_1 + (b_n -  \lfloor b_n \rfloor) e'_2 + c_n e'_3$ 
 also is a sequence of integer vectors such that for each integer $n$, $p_{e'_3}(x'_n) = r_n$. Moreover this sequence is bounded. Thus $(x'_n)_n \in  \mathbb Z$ takes a finite number of values, which enforces that the sequence $(r_n)_n \in  \mathbb Z$ also takes a finite number of values, which is a contradiction. 
\end{proof}

For KSPM(3), after a short transitional of logarithmic length,  hypotheses of theorem \ref{corollary:peak} are verified , and the study of avalanches can be turned into a pseudo linear process. Note that a trivial framing of the maximal non-empty column $e(N)$ of a fixed point with $N$ grains shows that $e(N)$ is in $\Omega(\sqrt{N})$. As a consequence, pseudo local process stands for the asymptotically complete behavior of avalanches.\\

Unfortunately, the approach above does not hold for $D >3$.  The main reason is that, for $D = 3$ unfired columns induce a very particular and \emph{periodic} prefix ($(2, 0)^j$) on configurations. From $D = 4$, the structure of such a possible prefix is more complex and we did not yet get a tractable characterization of those prefixes. 

\section{Perspectives}

In this paper we described avalanches as pseudo local processes from a certain column $l$.

We proved this column to be logarithmic in the number of grains $N$ for KSPM(3), leading to an asymptotically complete description of avalanches in that case. Simulations for other parameter $D$ suggests that the same outcome also holds.

The pseudo local process description involves some properties on avalanches, which we hope will be useful toward the study of fixed points shape. For an avalanche $s$, a particularly interesting consequence is that two successive fixed points are equal from $l+D-1$ to $(\max s) -1$, which hints that next avalanche reaching this part of the configuration may have a similar behavior. This would lead to a knowledge on the likeness of successive avalanches and therefore a foresee on the shapes of fixed points. Further work may concentrate on this point, where the main purpose is to go ahead iterating evolution rules, and to describe fixed points with a plain formula.

\bibliographystyle{plain}
\bibliography{biblio}

\end{document}